\documentclass[reprint,pra,twocolumn,showpacs]{revtex4-1}

\usepackage{amsthm}
\usepackage{amsmath}
\usepackage{latexsym}
\usepackage{amsfonts}
\usepackage{amssymb}
\usepackage{color}
\usepackage{bbm,dsfont}
\usepackage{graphicx}
\usepackage{hyperref}
\usepackage{enumerate}
\usepackage{comment}

\newtheorem{proposition}{Proposition}

\theoremstyle{definition}

\newcommand{\R}{\mathbb{R}} %real
\newcommand{\C}{\mathbb{C}} %complex
\newcommand{\real}{\mathbb R} %real

\newcommand{\mo}[1]{\left| #1 \right|} 
\newcommand{\abs}{\mo} %modulus

\newcommand{\e}{{\rm e}} %exponential

\newcommand{\rr}{\mathcal{R}}

\newcommand{\hi}{\mathcal{H}} %Hilbert space
\newcommand{\hh}{\mathcal{H}} %Hilbert space
\newcommand{\lhs}{\mathcal{L}_s(\hi)} %bounded seladjoint linear operators
\newcommand{\elles}[1]{\mathcal{L}_s \left( #1 \right)} %bounded sa operators on #1
\newcommand{\tr}[1]{{\rm tr}\left[#1\right]} %trace
\newcommand{\id}{\mathbbm{1}} %identity operator
\renewcommand{\rho}{\varrho}

\newcommand{\rank}[1]{\mathrm{rank}(#1)} %rank
\newcommand{\rankd}[1]{\mathrm{rank}_{\downarrow}(#1)} %rank
\renewcommand{\Re}{{\rm Re}\,}
\renewcommand{\Im}{{\rm Im}\,}

\newcommand{\Ao}{\mathsf{A}}%generic observable
\newcommand{\Bo}{\mathsf{B}}%generic observable
\begin{document}\setlength{\arraycolsep}{2pt}

\title[Expanding the Principle of Local Distinguishability]{Expanding the Principle of Local Distinguishability}

\author{Claudio Carmeli}
\email{claudio.carmeli@gmail.com}
\affiliation{DIME, Universit\`a di Genova, Via Magliotto 2, I-17100 Savona, Italy}

\author{Teiko Heinosaari}
\email{teiko.heinosaari@utu.fi}
\affiliation{Turku Centre for Quantum Physics, Department of Physics and Astronomy, University of Turku,  FI-20014, Finland}

\author{Jussi Schultz}
\email{jussi.schultz@gmail.com}
\affiliation{Dipartimento di Matematica, Politecnico di Milano, Piazza Leonardo da Vinci 32, I-20133 Milano, Italy}
\affiliation{Turku Centre for Quantum Physics, Department of Physics and Astronomy, University of Turku, FI-20014, Finland}

\author{Alessandro Toigo}
\email{alessandro.toigo@polimi.it}
\affiliation{Dipartimento di Matematica, Politecnico di Milano, Piazza Leonardo da Vinci 32, I-20133 Milano, Italy}
\affiliation{I.N.F.N., Sezione di Milano, Via Celoria 16, I-20133 Milano, Italy}

\begin{abstract}
The principle of local distinguishability states that an arbitrary physical state of a bipartite system can be determined by the combined statistics of local measurements performed on the subsystems. A necessary and sufficient requirement for the local measurements is that each one must be able to distinguish between all pairs of states of the respective subsystems. We show that if the task is changed into the determination of an arbitrary bipartite pure state, then at least in certain cases it is possible to restrict to local measurements which can distintinguish all pure states but not all states. Moreover, we show that if the local measurements are such that the purity of the bipartite state  can be verified from the statistics without any prior assumption, then in these special cases also this property is carried over to the composite measurement. 
These surprising facts give evidence that the principle of local distinguishability may be expanded beyond its usual applicability.
\end{abstract}


\maketitle

%%%%%%%%%%%%%
\section{Introduction}
%%%%%%%%%%%%%
Quantum theory is an example of a physical theory which satisfies  the \emph{principle of local distinguishability}.
This means that if two states of a composite system are different, then they can be distinguished using the combined statistics of some appropriate local measurements on the component systems. 
This feature of quantum theory has both practical and foundational relevance. From the practical point of view, local measurements are obviously easier to implement than global measurements. On the foundational side, local distinguishability has been considered such an important feature  
that it has been taken as an axiom in several derivations of quantum theory \cite{Barrett07,ChDaPe11,MaMu11}.

The principle of local distinguishability is a statement concerning arbitrary states of composite physical systems. 
However, it is completely reasonable to ask if the same principle holds when only pure states are considered, and this is precisely the focus of this paper. To be more explicit, suppose that Alice and Bob both perform local measurements that are capable of identifying an unknown pure state among all pure states on their respective components; see Fig. \ref{fig:local}. Are they then able, using the combined statistics of those measurements, to identify an unknown pure state among all pure states of the composite system?

Since pure states represent the states of maximal information for a system, this variation of the principle of local distinguishability is of foundational interest.  
The main practical motivation comes from the fact that prior information, in this case the purity of the unknown state, can be exploited to drastically reduce the amount of resources needed for state tomography \cite{HeMaWo13, ChDaJietal13, CaHeScTo14}.  
Specifically, for pure state determination the minimum number of measurement outcomes needed to succeed in the task reduces from a quadratic (in the dimension of the system) to a linear expression. The validity of this expanded principle of local distinguishability would therefore imply that the experimenter could take advantage of not only the simpler setup coming from the locality of the measurements, but also the sufficiency of the restricted resources.  

In this paper we show that Alice and Bob can succeed in their task of pure state determination at least in two cases: (i) if at least one party, Alice or Bob, can distinguish between {\em all} states of their respective subsystem, or (ii) if at least one of the subsystems is either a qubit or a qutrit. We then show that if Alice and Bob possess measurements which are in addition capable of verifying the purity of their systems from the statistics, then in the above cases the composite measurement also has this property. Finally, we generalize case (ii) to multipartite systems consisting of qubits and qutrits, and use this to obtain a special class of measurements on higher dimensional systems, for which the expanded principle of local distinguishability is valid.

%%%%%%%%%%%%%%%%%%%%%%
\section{Pure state informational completeness}
%%%%%%%%%%%%%%%%%%%%%%
Recall that a measurement is called \emph{informationally complete} if any two different states can be distinguished from the outcome statistics \cite{Prugovecki77}. 
Mathematically such a measurement is described by a positive operator valued measure (POVM) $\Ao$ such that the elements $\Ao(x)$ span the real vector space $\lhs$ of selfadjoint operators on the Hilbert space $\hh$ of the system \cite{Busch91,SiSt92}. (We always assume  $\dim\hi<\infty$.)
This is equivalent to the requirement that the expectation value of any observable $O$ can be written as a linear combination of the probabilities $\varrho^\Ao(x) = \tr{\varrho \Ao(x)}$,
\begin{equation}\label{eq:exp}
\langle O\rangle = \sum_x \alpha_x \varrho^\Ao(x) \, .
\end{equation}

As a variation of informational completeness, we say that a measurement is \emph{pure state informationally complete} if any two different \emph{pure} states give different measurement outcome statistics \cite{BuLa89}. 
In order to formulate a mathematical criterion for this property, we first define $\rr(\Ao)$ to be the real linear span of the operators $\Ao(x)$ of a POVM $\Ao$, i.e., $\rr(\Ao) = \{ \sum_x r_x \Ao(x) :r_x\in\real\}$. 
In physical terms, $\rr(\Ao)$ is the set of those observables for which we can calculate the expectation value in the form \eqref{eq:exp}. 
As said before, $\rr(\Ao)=\lhs$ if and only if $\Ao$ is informationally complete.

\begin{figure}
\includegraphics[width=7.5cm]{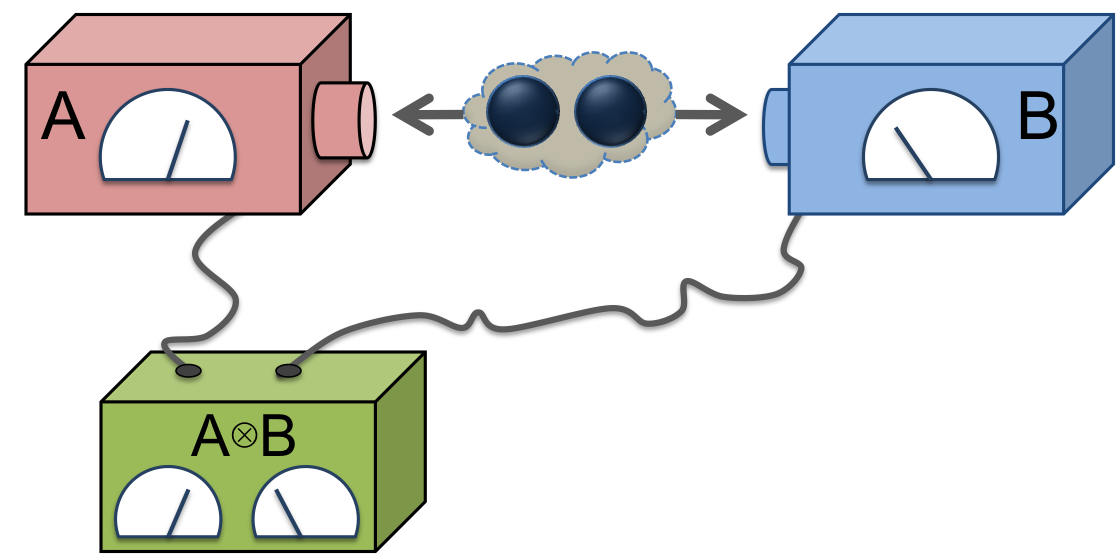} 
\caption{(Color online) For a bipartite quantum system, any local measurements $\Ao$ and $\Bo$ performed on the subsystems can be combined to yield a measurement $\Ao\otimes\Bo$ on the composite system. The question then is if $\Ao\otimes\Bo$ can distinguish between all bipartite pure states whenever $\Ao$ and $\Bo$ can distinguish between all pure states of the respective subsystems.}
\label{fig:local}
\end{figure}

We denote by $\rr(\Ao)^\perp\subset\lhs$ the orthogonal complement of $\rr(\Ao)$ with respect to the Hilbert-Schmidt inner product $\langle S\mid T\rangle =\tr{ST}$. 
It is easy to verify that two states $\varrho_1$ and $\varrho_2$ are indistinguishable by a measurement of $\Ao$  if and only if $\varrho_1-\varrho_2 \in \rr(\Ao)^\perp$. 
Therefore, the pure state informational completeness of $\Ao$ is equivalent to the condition that $\rr(\Ao)^\perp$ does not contain any matrix of rank 2 \cite{HeMaWo13}. 
Indeed, since $\rr(\Ao)$ contains the identity $\id$, the elements of $\rr(\Ao)^\perp$ are traceless and therefore the existence of such a matrix would yield two indistinguishable pure states via its spectral decomposition. 

The two properties of informational completeness and pure state informational completeness are inequivalent for all systems whose Hilbert space is at least three dimensional \cite{HeMaWo13}. 
For instance, an informationally complete measurement for a qutrit system requires at least $9$ outcomes, but a pure state informationally complete measurement can have minimally $8$ outcomes 
 (for a concrete example, see \cite{CaHeTo13}).
The fact that the two concepts are equivalent for qubit systems follows easily from the previously mentioned rank condition: any element of $\rr(\Ao)^\perp$ is a selfadjoint and traceless $2\times 2$ matrix, hence it is either zero or has rank $2$; therefore, $\Ao$ is pure state informationally complete if and only if $\rr(\Ao)^\perp = \{0\}$, that is, $\Ao$ is informationally complete.

\smallskip

%%%%%%%%%%%%%%%%%%%%%%
\section{Expanding local distinguishability}\label{sec:expanding}
%%%%%%%%%%%%%%%%%%%%%%
Let us turn to a setting where the unknown state is a joint state of a composite system, and two local measurements are performed on its subsystems; see Fig. \ref{fig:local}. 
Alice and Bob thus measure some POVMs $\Ao$ and $\Bo$ acting on the Hilbert spaces $\hh_A$ and $\hh_B$ of the subsystems, respectively. The {composite} measurement is then described by the tensor product POVM $(\Ao\otimes\Bo)(x,y) = \Ao(x)\otimes\Bo(y)$ acting on $\hh_A\otimes\hh_B$.
Our main question can be formulated as follows:
\begin{quote}
\emph{If $\Ao$ and $\Bo$ are  pure state informationally complete, does it follow that also $\Ao\otimes\Bo$ is pure state informationally complete?}
\end{quote}
We first note that the converse is true: if the composite measurement $\Ao\otimes \Bo$ is pure state informationally complete, then so are both of the components. 
Indeed, if it were the case that, say, $\Ao$ could not distinguish some pair of distinct pure states $\varrho_1$ and $\varrho_2$, then any pure state $\sigma$ of Bob's component would yield distinct pure states $\varrho_1\otimes \sigma$ and $\varrho_2\otimes \sigma$ which would be indistinguishable by $\Ao\otimes \Bo$.

In order to get a grasp of the problem at hand, we need to understand the structure of the complement space $\rr(\Ao\otimes\Bo)^\perp $. 
First note that since $\rr(\Ao\otimes\Bo) = \rr(\Ao)\otimes\rr(\Bo)$, each of the three orthogonal subspaces $\rr(\Ao)^\perp \otimes \rr(\Bo)$, $ \rr(\Ao) \otimes \rr(\Bo)^\perp$, and   $ \rr(\Ao)^\perp \otimes \rr(\Bo)^\perp$ is contained in $\rr(\Ao\otimes\Bo)^\perp $, and hence so is their direct sum. By dimension counting it can be verified that this direct sum is actually equal to the complement space. 
Furthermore, since  $\elles{\hi_A} = \rr(\Ao)\oplus\rr(\Ao)^\perp$ and similarly for $\elles{\hi_B}$, we have the
 following expressions:
\begin{align*}
\rr(\Ao\otimes\Bo)^\perp &= (\elles{\hi_A}\otimes \rr(\Bo)^\perp) \oplus (\rr(\Ao)^\perp \otimes \rr(\Bo)) \\
& = (\rr(\Ao)^\perp \otimes \elles{\hi_B}) \oplus (\rr(\Ao) \otimes \rr(\Bo)^\perp) 
\end{align*}
These equations will be used repeatedly in the rest of the paper.

\smallskip

%%%%%%%%%%%
\subsection{Informational completeness on one side}
%%%%%%%%%%%
We begin our investigation by considering the special case where one party, say Alice, can perform an informationally complete measurement. We then have $\rr (\Ao)^\perp =\{ 0 \}$, which implies that 
\begin{equation}\label{eq:ic_alice}
\rr (\Ao\otimes \Bo)^\perp = \mathcal{L}_s(\hh_A)\otimes \rr (\Bo)^\perp
\end{equation}
by our previous observation. In particular, if also Bob performs an informationally complete measurement, then $\rr (\Ao\otimes \Bo)^\perp=\{0\}$, which confirms the usual form of local distinguishability:  $\Ao\otimes \Bo$ is informationally complete if both $\Ao$ and $\Bo$ are such. 
Interestingly, the next result shows that also pure state informational completeness on Bob's side carries over to the composite measurement.

\begin{proposition}\label{prop:ic_alice}
Let $\Ao$ be an informationally complete and $\Bo$ a pure state informationally complete measurement. Then $\Ao\otimes\Bo$ is pure state informationally complete.
\end{proposition}
\begin{proof}
As we noted earlier, in order to prove the pure state informational completeness of $\Ao\otimes\Bo$, we need to show that the rank of any nonzero matrix $T$ in $\rr (\Ao\otimes \Bo)^\perp $  is at least $3$. 
From \eqref{eq:ic_alice} we see that a nonzero $T\in\rr (\Ao\otimes \Bo)^\perp $ can be written in a block form as 
\begin{align*}
T= \left(\begin{array}{cccc}
T_{11} & T_{12} & \cdots & T_{1,d_A} \\
T_{21} & T_{22} & \cdots & T_{2,d_A} \\
\vdots & \vdots & \ddots & \vdots \\ 
T_{d_A,1} & T_{d_A,2} & \cdots & T_{d_A,d_A} 
\end{array}\right) \, , 
\end{align*}
where $d_A=\dim \hi_A$ and each $T_{jk}$ is an element of the complex linear span of $\rr(\Bo)^\perp$  satisfying $T_{jk}^* = T_{kj}$. 

Firstly, suppose that $T_{jj}\neq 0$ for some $j=1,\ldots,d_A$.
We have $\rank{T_{jj}} \geq 3$ since $T_{jj}\in \rr(\Bo)^\perp$ and $\Bo$ is pure state informationally complete.  
This implies that $\rank{T} \geq 3$ as $\rank{T}\geq\rank{T_{jj}}$.

Secondly, suppose that $T_{jj}=0$ for all $j=1,\ldots,d_A$, thus $T_{jk}\neq 0$ for  some $j\neq k$.
Since $T_{jk}$ need not be selfadjoint, it may be not in $\rr(\Bo)^\perp$.
However, the real part $\Re T_{jk} = (T_{jk} + T_{jk}^*)/2$ and the imaginary part $\Im T_{jk}=(T_{jk} - T_{jk}^*)/2i$ are selfadjoint and therefore elements of $\rr(\Bo)^\perp$. 
Since $T_{jk}\neq 0$, we have $\Re T_{jk}\neq 0$ or $\Im T_{jk}\neq 0$.
It thus suffices to show that $\rank{T}\geq \max\{\rank{\Re T_{jk}}, \rank{\Im T_{jk}}\}$.
To see this, we denote
\begin{align*}
\widetilde{T} = \left(
\begin{array}{cc}
0 & T_{jk} \\
T_{jk}^\ast & 0
\end{array}
\right) \, .
\end{align*}
As $\widetilde{T}$ is a submatrix of $T$, we have $\rank{T} \geq \rank{\widetilde{T}}$.
We further observe that 
\begin{align*}
\widetilde{T}= V {\widetilde{T}}_0 V^*
\end{align*}
where
\begin{align*}
\widetilde{T}_0& = \left(
\begin{array}{cc}
\Im T_{jk} & \Re T_{jk} \\
\Re T_{jk} & -\Im T_{jk}
\end{array}
\right)
\end{align*}
and $V$ is the unitary block matrix
$$
V = \frac{1}{\sqrt{2}} \left(
\begin{array}{cc}
\id & -i\id \\
-i\id & \id
\end{array}
\right) \, .
$$
Since we have $\rank{\widetilde{T}} = \rank{\widetilde{T}_0}$ and moreover $\rank{\widetilde{T}_0}\geq \max\{\rank{\Re T_{jk}},\rank{\Im{T_{jk}}}\}$, this implies\begin{align*}
\rank{T}\geq \mathrm{max}\{ \rank{\Re T_{jk}},\rank{\Im T_{jk}} \} \geq 3 \, .
\end{align*} 
\end{proof}

\smallskip

%%%%%%%%%%%%%%%%%
\subsection{Qutrit on one side}
%%%%%%%%%%%%%%%%%
We now wish to drop the assumption of informational completeness for Alice's measurement, and assume only that she can distinguish all pure states. However, as the dimensions of the systems increase, so does the complexity of the space $\rr (\Ao\otimes\Bo)^\perp$. As a result, an exhaustive answer to our question still remains to be found. 

In the special case that Alice's system is a qutrit, the structure of $\rr (\Ao\otimes \Bo)^\perp$ is manageable. 
Suppose that $\Ao$ is pure state informationally complete. 
We have then two possibilities:  (i) $\rr(\Ao)^\perp=\{0\}$, in which case $\Ao$ is actually informationally complete with respect to all states, or (ii)  $\rr(\Ao)^\perp = \R \,S = \{ rS:r\in\real\}$ for some invertible matrix $S$.
It can be shown that these are the only possibilities for pure state informational completeness for a qutrit system \cite{HeMaWo13} .
The case (i) was already treated earlier, so we concentrate on (ii).
In that case we have
\begin{align}\label{eq:prodrank3}
\rr (\Ao\otimes \Bo)^\perp = (S\otimes \elles{\hh_B})\oplus (\rr(\Ao)\otimes \rr(\Bo)^\perp) \, . 
\end{align}
This simplified structure allows us to prove the next result.

\begin{proposition}\label{prop:qutrit_pic}
Let $\dim\hh_A=3$. Then $\Ao\otimes \Bo$ is pure state informationally complete if and only if $\Ao$ and $\Bo$ are pure state informationally complete.
\end{proposition}

\begin{proof}
One implication has already been established at the beginning of this section. Let us then prove the other one.
Assume that $\Ao$ and $\Bo$ are both pure state informationally complete. We may assume that neither of them is informationally complete since this case was already in Proposition \ref{prop:ic_alice}. 
Let $\rr(\Ao)^\perp = \R \,S $ with a full rank matrix $S$. It is not restrictive to assume that 
\begin{align*}
S= \left(\begin{array}{ccc}
s_{1} & 0 & 0\\
0 & s_{2} & 0 \\
0 & 0 & s_{3}
\end{array}\right)
\end{align*}
where $s_i\in\R$ with $s_1s_2s_3\neq 0$ and $s_1 + s_2 + s_3 = 0$. 
By \eqref{eq:prodrank3} any  $T\in\rr(\Ao\otimes \Bo)^\perp$ can be written as a block matrix 
\begin{align*}
T & =\left(
\begin{array}{ccc}
T_{11} & T_{12} & T_{13}\\
T_{21} & T_{22} & T_{23}\\
T_{31} & T_{32} & T_{33}
\end{array}
\right)
\end{align*}
where $T_{jj}  = s_j L + R_j $ with $R_j\in\rr(\Bo)^\perp$  and $ L\in\elles{\hh_B}$, and each $T_{jk}$ with $j\neq k$ is from the complex linear span of $\rr(\Bo)^\perp$ and satisfies $ T_{jk}^* = T_{kj}$.

Choose the unitary block matrix
\begin{align*}
U = \frac{1}{\sqrt{6}}\,\left(
\begin{array}{ccc}
\sqrt{2}\,\id & \sqrt{2}\e^{i\alpha}\,\id & \sqrt{2}\e^{i\beta}\,\id\\
\sqrt{3}\,\id & 0 & -\sqrt{3}\e^{i\beta}\,\id\\
\id & -2\e^{i\alpha}\,\id & \e^{i\beta}\,\id
\end{array}
\right) \, .
\end{align*}
Then $\widetilde{T} = UTU^\ast$ is such that
\begin{align*}
\widetilde{T}_{11} & = \frac{1}{3}\, \sum_{i=1}^3 R_i + \frac{2}{3}\, [\Re T_{21} \cos\alpha - \Im T_{21} \sin\alpha\\
& + \Re T_{31} \cos\beta - \Im T_{31} \sin\beta\\
& + \Re T_{32} \cos(\beta-\alpha) - \Im T_{32} \sin(\beta-\alpha)] \,.
\end{align*}
If now $\widetilde{T}_{11} \neq 0$ for some $\alpha,\beta$, then $\rank{T} = \rank{\widetilde{T}} \geq \rank{ \widetilde{T}_{11}} \geq 3$  since $\widetilde{T}_{11} \in \rr(\Bo)^\perp$ and $\Bo$ is pure state informationally complete. 
Suppose instead that  $\widetilde{T}_{11} = 0$ for all $\alpha,\beta$. 
Then by the linear independence of the functions $1$, $\cos\alpha$, $\sin\alpha$, $\cos\beta$, $\sin\beta$, $\cos(\beta-\alpha)$ and $\sin(\beta-\alpha)$ in the two variables $\alpha$ and $\beta$, we have $\sum_{j=1}^3 R_j = 0$ and $T_{jk} = 0$ for $j\neq k$.
Thus, if $T_{jj} \neq 0$ for all $j$, then $\rank{T}\geq 3$ trivially. 
If otherwise $T_{jj} = 0$ for some $j$, then $L=-R_j/s_j$, and $T_{kk} = R_k - (s_k/s_j)R_j \in \rr(\Bo)^\perp$ is $0$ or has rank at least $3$ by the pure state informational completeness of $\Bo$. 
Thus, $T = 0$ or $\rank T\geq \rank{T_{kk}}\geq 3$ for some $k$. 
In conclusion, $\Ao\otimes \Bo$ is pure state informationally complete.
\end{proof}

%%%%%%%%%%%%%
\section{Verifying purity}
%%%%%%%%%%%%%
In order for a pure state informationally complete measurement to be useful, the experimenter must know {\em a priori} that the system is in a pure state. Since the purity of the state is a very delicate property, it would be desirable to be able to verify this premise directly from the statistics. Furthermore, as the ultimate goal is the determination of the state after verifying the premise, one actually needs a measurement that can distinguish an arbitrary pure state from any other state, pure or mixed.

We say that a measurement is \emph{verifiably pure state informationally complete}   if the measurement outcome statistics of a pure state is different from the outcome statstics of any other state \cite{ChDaJietal13}.  Note that the minimal number of measurement outcomes needed for verifiable pure state informational completeness also scales linearly with the dimension of the system \cite{ChDaJietal13}. A mathematical criterion for this property  can be formulated as follows \cite{CaHeScTo14}:  an observable $\Ao$ is verifiably pure state informationally complete if and only if every nonzero $T\in\rr(\Ao)^\perp$ has 
\(
\rankd{T}\geq 2
\), 
where 
$$
\rankd{T} = \min\{\rank{T+\abs{T}},\,\rank{T-\abs{T}}\}
$$
is the minimum between the number of strictly positive and strictly negative eigenvalues of $T$.

As an immediate consequence, in the qubit and qutrit cases verifiable pure state informational completeness is equivalent to informational completeness. Indeed, in these cases the above criterion implies that $\Ao$ is verifiably pure state informationally complete if and only if $\rr(\Ao)^\perp = \{0\}$.

For composite systems verifiable pure state informational completeness behaves in a similar way as pure state informational completeness. Indeed, as done in the previous section, one can easily  prove   that  a necessary condition for $\Ao\otimes\Bo$ to be verifiably pure state informationally complete is that both $\Ao$ and $\Bo$ are such.
Moreover, the next two results are the complete analogues of Propositions~\ref{prop:ic_alice} and \ref{prop:qutrit_pic}.

\begin{proposition}
\label{prop:ICVPIC}
Let $\Ao$ be an informationally complete measurement.
If $\Bo$ is  verifiably pure state informationally complete, then $\Ao\otimes\Bo$ is  verifiably pure state informationally complete.
\end{proposition}

\begin{proof}
It is not restrictive to assume that $d_B = \dim\hh_B \geq 4$, as otherwise the verifiable pure state informational completeness of $\Bo$ implies that $\Bo$ is informationally complete, and the claim is the usual local distinguishability.

As in the proof of Proposition~\ref{prop:ic_alice}, 
 each element $T\in \rr(\Ao\otimes\Bo)^\perp$ can be written as a $d_A\times d_A$ block matrix 
with entries $T_{ij}$ in  the complex linear span of the set $\rr(\Bo)^\perp$, and such that $T_{ij}^\ast=T_{ji}$.
We now assume that $T\neq 0$ and show that $\rankd{T} \geq 2$. By the criterion stated above, this will imply the verifiable pure state informational completeness of $\Ao\otimes\Bo$.
 
If $T_{ii}\neq 0$ for some $i$, then $T_{ii}\in\rr(\Bo)^\perp\setminus\{0\}$, hence $\rankd{T_{ii}}\geq 2$. 
We claim that in this case $\rankd{T} \geq 2$.
Indeed, $T_{ii}$ is a compression of $T$ to a $d_B$-dimensional subspace of $\hh_\Ao\otimes\hh_\Bo = \C^{d_A d_B}$.
The Cauchy interlacing theorem (see \cite[Corollary III.1.5]{Bhatia}) states that, for each $j=1,2,\ldots, d_B$, 
\[
\lambda_j(T)\geq \lambda_j (T_{ii})\geq \lambda_{j+d_B(d_A-1)}(T) \,,
\]
where $\lambda_j(T)$, $\lambda_j (T_{ii})$ are the eigenvalues of $T$ and $T_{ii}$, respectively, possibly repeated according to their multiplicities and listed in decreasing order.
The condition $\rankd{T_{ii}}\geq 2$ means that $T_{ii}$ has at least $2$ strictly positive and $2$ strictly negative eigenvalues. 
Hence, the first of the Cauchy interlacing inequalities yields
\[
\lambda_j(T) \geq \lambda_j(T_{ii})>0 \quad \textrm{ for } j=1,2
\]
and the second gives
\[
0>\lambda_j(T_{ii}) \geq \lambda_{j+d_B(d_A-1)}(T) \quad \textrm{ for } j=d_B -1,d_B\,.
\]
In other words, $T$ also has at least $2$ strictly positive and $2$ strictly negative eigenvalues, that is, $\rankd{T}\geq 2$.

Finally it remains to consider the case in which  $T_{ii} = 0$ for all $i=1,\ldots,d_A$.  
As $T\neq 0$, we have $T_{ij} \neq 0$ for some $i\neq j$. 
By relabeling the entries of $T$ if necessary, we can assume that $i=1$ and $j=2$. Consider then the upper-left square minor
$$
\widetilde{T} = \left(
\begin{array}{cc}
0 & T_{12} \\
T_{12}^\ast & 0
\end{array}
\right)
= U \widetilde{T}_+ U^* = V \widetilde{T}_- V^*
$$
where
\begin{align*}
\widetilde{T}_+ & = \left(
\begin{array}{cc}
\Re T_{12}  & i\,\Im T_{12} \\
-i\,\Im T_{12} & -\Re T_{12}
\end{array}
\right) \\
\widetilde{T}_- &= \left(
\begin{array}{cc}
\Im T_{12} & \Re T_{12} \\
\Re T_{12} & -\Im T_{12}
\end{array}
\right)
\end{align*}
and $U$ and $V$ are the unitary block matrices
$$
U = \frac{1}{\sqrt{2}} \left(
\begin{array}{cc}
\id & -\id \\
\id & \id
\end{array}
\right)
\qquad
V = \frac{1}{\sqrt{2}} \left(
\begin{array}{cc}
\id & -i\id \\
-i\id & \id
\end{array}
\right) \,.
$$
By another application of the Cauchy interlacing inequalities,
\begin{align*}
\rankd{T} & \geq \rankd{\widetilde{T}} = \rankd{\widetilde{T}_+} \geq \rankd{\Re T_{12}}
\end{align*}
and
\begin{align*}
\rankd{T} & \geq \rankd{\widetilde{T}} = \rankd{\widetilde{T}_-} \geq \rankd{\Im T_{12}} \,.
\end{align*}
Both $\Re T_{12}$ and $\Im T_{12}$ belong to $\rr(\Bo)^\perp$ and at least one of them is nonzero as $T_{12}$ is nonzero.
Therefore, $\rankd{T}\geq 2$.
\end{proof}

\begin{proposition}\label{prop:qutrit_vpic}
Let $\dim\hi_A=3$. Then $\Ao\otimes \Bo$ is verifiably pure state informationally complete if and only if $\Ao$ and $\Bo$ are verifiably pure state informationally complete.
\end{proposition}

\begin{proof}
As already noticed, necessity is easy. On the other hand, in dimension $3$ verifiable pure state informational completeness is equivalent to informational completeness, hence sufficiency follows by Proposition \ref{prop:ICVPIC}.
\end{proof}

%%%%%%%%%%%%
\section{Extension to multipartite systems}
%%%%%%%%%%%%
Since in Proposition \ref{prop:qutrit_pic} no assumption regarding the dimension of Bob's system was made, we can use it to  obtain an extension to the multipartite case. Suppose that we have $N$ quantum systems, each of which is either a qubit or a qutrit, and suppose that we have the corresponding $N$  pure state informationally complete measurements described by the POVMs $\Ao_i$. We denote  
$$
\Ao ({\bf x}) = \bigotimes_{i=1}^N \Ao_i (x_i), \qquad {\bf x} = (x_1,\ldots, x_N).
$$
We consider the splitting of the POVM into two parts
$$
\Ao = \Ao_1 \otimes \left( \bigotimes_{i=2}^N \Ao_i\right).
$$ 
Propositions \ref{prop:ic_alice} and \ref{prop:qutrit_pic} tell us that $\Ao$ is pure state informationally complete if and only if each of the two factors is such (recall that a pure state informationally complete qubit measurement is necessarily informationally complete). 
By induction, we may conclude that actually $\Ao$ is pure state informationally complete if and only if each component $\Ao_i$ is such. 
. 

This multipartite extension also gives a class of POVMs on higher dimensional systems for which the pure state informational completeness of a pair of POVMs carries over to their tensor product. Namely, suppose that the only prime components of the dimensions $d_A$ and $d_B$ are $2$ and $3$, that is  $d_A = 2^{n_A}3^{m_A}$ and $d_B = 2^{n_B} 3^{m_A}$. This means that we can write 
$$
\hh_A = \left(\bigotimes_{j=1}^{n_A} \C^2 \right)\otimes \left(\bigotimes_{k=1}^{m_A} \C^3\right),
$$
and similarly for $\hh_B$. Suppose now that $\Ao$ and $\Bo$ are pure state informationally complete POVMs which also factorize into tensor products
$$
\Ao = \bigotimes_{j=1}^{n_A + m_A} \Ao_j, \qquad \Bo = \bigotimes_{k=1}^{n_B + m_B} \Bo_k 
$$
where each component acts on $\C^2$ or $\C^3$. Each component must be pure state informationally complete, and therefore by our multipartite result, so is $\Ao\otimes\Bo$. 

%%%%%%%%%%%%
\section{Conclusions} 
%%%%%%%%%%%
An entangled pure state of a composite system has mixed reduced states.
For this reason, a local measurement not distinguishing mixed states may seem quite useless for quantum tomography on a bipartite system.
However, as we have shown, if Alice can implement an informationally complete measurement and Bob can distinguish all pure states with his measurement, then they can together distinguish all pure states of the composite system. Furthermore, in the case that Alice's system is a qutrit she can also choose to perform merely a pure state informationally complete measurement while still maintaining the ability to distinguish all pure bipartite states. In a similar manner, if Alice and Bob are able to verify the purity of the state of their respective systems, then in the above cases the composite measurement also has this property.

Quantum theory has two quite opposite features: information is stored globally (entanglement), but can be retrieved with local measurements (the principle of local distinguishability). 
It is an interesting question if this balance can be properly quantified and if it is unique to quantum theory. We believe that our investigation on the principle of local distinguishability can stimulate a new direction in the axiomatization of quantum theory, but also may help to design quantum tomography schemes with reduced resources.

%%%%%%%%%%
\section*{Acknowledgments}
%%%%%%%%%%
The authors are grateful to Tom Bullock and M\'ario Ziman for their comments on an earlier version of this manuscript.
TH acknowledges financial support from the Academy of Finland (grant no 138135). JS and AT acknowledge financial support from the Italian Ministry of Education, University and Research (FIRB project RBFR10COAQ).

%\bibliographystyle{unsrt}
%\bibliographystyle{alpha}
%\bibliographystyle{abbrv}
%\bibliographystyle{acm}
%\bibliographystyle{Nabbrv}
%\bibliographystyle{amsplain}
%\bibliographystyle{phcpc}

%%%%%%%%%%
%%%%%%%%%%
\end{document}